\documentclass[a4paper,11pt]{article}
\usepackage[scale={0.75,0.9},centering,includeheadfoot]{geometry}
\usepackage{geometry}
\usepackage{setspace}
\usepackage[latin1]{inputenc} 
\usepackage[british]{babel}
\usepackage{color} 
\usepackage{amsmath,amsfonts,amssymb,amsthm}
\usepackage{graphicx}
\usepackage{booktabs,tabularx} 
\usepackage{floatrow}

\usepackage{algorithm,algorithmic}
\usepackage{lineno}

\newcommand{\mv}[1]{{\boldsymbol{\mathrm{#1}}}}
\newcommand{\E}{  \mathbb{E} }
\newcommand{\V}{  \mathbb{V} }

\newcommand{\proper}{\mathsf}

\newcommand{\plN}{\proper{LN}}

\newtheorem{Theorem}{Theorem}

\theoremstyle{definition} 
\newtheorem{Definition}{Definition} 
 
\theoremstyle{remark} 

\graphicspath{{../Figures/},{../}}
\newcommand{\Lund}{Statistics, Lund university, Sweden}
\newcommand{\Gothenburg}{Department of Biological and Environmental Sciences,\\ University of Gothenburg, Sweden}

\begin{document}

\begin{center}
\textbf{\textsf{\huge
Estimating the unobservable moose - converting index to population size using a Bayesian Hierarchical state space model}}\\
\vspace{5mm}
{\Large \scshape{Jonas Wallin$^1$ and Kjell Wallin$^2$}}

\vspace{3mm}

\textit{$^1$\Lund}\\
\textit{$^2$\Gothenburg}\\ 

\vspace{3mm}
\begin{minipage}{0.9\textwidth}
{\small
{\bf ABSTRACT.

Indirect information on population size, like pellet counts or volunteer counts, is the main source of information in most ecological studies and applied population management situations. Often, such observations are treaded as if they were actual measurements of population size. This assumption results in incorrect conclusions about a population's size and its dynamics. We propose a model with a temporal varying link, denoted countability, between indirect observations and actual population size.
We show that, when indirect measurement has high precision (for instance many observation hours) the assumption of temporal varying countability can have a crucial effect on the estimated population dynamic.
We apply the model on two local moose populations in Sweden. The estimated population dynamics is found to explain 30-50 percent of the total variability in the observation data; thus, countability accounts for most of the variation. This unreliability of the estimated dynamics has a substantial negative impact on the ability to manage populations; for example, reducing (increasing) the number of animals that needs to be harvested in order to sustain the population above (below) a fixed level.
Finally,  large difference in countability between two study areas implies a substantial spatial variation in the countability; this variation in itself is highly worthy of study.
}
\begin{flushleft}
\end{flushleft}
}
\end{minipage}

\vspace{3mm}
\end{center}

\section{Introduction}

Management of any wild animal population requires accurate information on its size and dynamics.
Unfortunately, this basic information is both notoriously difficult and expensive to acquire. In most situations only indices of population size, such as tracks, feces or numbers animals caught, are available  \cite{sutherland2006ecological}. Understanding the link between such indices and population size is fundamental to good management, as the link between these indices and the actual population size could well have a temporal and spatial variation of its own.
 In practice, this issue is often ignored and the indices are treated as true population measurements; resulting in artifactual patterns in estimated population sizes and unwanted consequence for management \cite{pollock2002large, williams2002analysis}.

In many applications, the actual population size is required, for instance when determining allowed take for sustainable harvest, estimation risk of extinction, or assessing spatial and temporal variability in growth rates \cite{engen1997harvesting, lande2003stochastic}. 
There are well established capture-recapture of distance methods for obtaining unbiased estimates of population size (\cite{borchers2002estimating, williams2002analysis}). However, such methods are often expensive, especially for large mammals, and in local management situations the use of population indices often appears to be the only option, at least on a regular basis.

In this article, we propose a state space model that combines indices and unbiased population estimates to make inference about population size, using a bias factor, which we denote countability, to link them together. The novelty of our approach is allowing for temporal variation of the bias factor between inidices and population sizes. 
If statistical models ignores the variation in countability this can lead to unreasonably large estimates of population size variability, that are simple an artifact of the ignoring the variation in countability.
 This increased variability in population size can have a large effect on management decisions, as is illustrated in our example.

We propose a Bayesian hierarchical  state space model, where the population is defined as a continuous stochastic process; the processes easily incorporates irregular spaced observation.  The observation part of the state space model has two different sources, index data and, at some  times, unbiased population estimates derived from capture-recapture or distance sampling.
Since the model has hyperparameters that need to be specified, we parameterize the priors in such a way that they have a clear interpretable meaning, or in cases where no prior knowledge is available the priors are chosen so as to prefer simpler model

The remainder of the article, is composed as follows: Section two is a short introduction to state space models, including the novel concept of 'effort homogeneous observation distributions', which is an important property for certain state space models applied to population dynamics. Section three describes in detail the hierarchical state space model developed in this article. In section four, the model is applied to two spatially separate moose populations. In section five we discuss the results and point to future research direction. 
\section{Theory}
The main modeling tool in this is article is the state space model. The goal of this section is to show the effect of effort on certain observation equations commonly used with state space models in ecology. To start, we formalize the state space model:
A state space model \cite{buckland2007}, is used to model an unobserved process, $N_t$, for times $t=1,\ldots,T$. Here, and throughout the paper, $N_t$ represent the population abundance. The general state space model can be described by two equations:
\begin{align}
\label{eq:sate1}
N_t \, | \, N_{t-1}&\sim f(n_t|N_{t-1}, \mv{\Theta}),  \\
\label{eq:sate2}
Y_t \, |\, N_{t}&\sim g(y_t| N_{t}, E_t, \mv{\Theta}),
\end{align}
where (\ref{eq:sate1}) is the state equation, and (\ref{eq:sate2}) is the observation equation. Here $Y\sim g$ implies that the density of the random variable $Y$ is $g$.
$Y_t$ represent the observations, $f$ and $g$ are distributions depending on the parameters $\mv{\Theta}$, and  $E_t$  which is the observation effort (e.g. number of hours spent looking for the animals or  length of the line transects).

Often, implicitly, the distributions in observation equations for state space model in ecology fall within the following class of distributions:
\begin{Definition}
Let $g(y| E)$ denote the distribution of an observation given effort $E$ and let $Y^i \sim g(y| E=a_ie)$ for some $a_i \geq 0$.
 Then $g(y| E)$ is effort-homogeneous if $Y^0$ has the same distribution as $Y^1 + Y^2$ for all  $e$ and all $a_i$ satisfying  $a_0=1$, $a_1 + a_2 = 1$.
\end{Definition}
The definition implies that if an observation distribution is effort-homogeneous, there is no gain in splitting the observation into smaller pieces (with respect to effort), and that each unit of effort contributes equally to the information of the population size. 
Common distributions that have this property are, for instance, the Poisson distribution with linear mean function, and the mark recapture model with a fixed capture probability over repeated visits.

A larger class of distributions which contains any effort homogeneous distribution is the weakly effort-homogeneous distributions.
\begin{Definition}
Let $g(y| E)$ denote the distribution of an observation given effort $E$.
 Then $g$ is weakly effort-homogeneous if it has finite variance for all $E$ and $\E[Y|  E=e]= \E[Y| E=a_1e] + \E[Y| E=a_2e]$ and  $\V[Y| E=e] = \V[Y |E=a_1e] + \V[Y| E=a_2e]$, for any $a_1,a_2$ where $a_1 + a_2 = 1$, and $a_1 \geq0, a_2 \geq 0$.
\end{Definition}
This assumption is typically implied in for instance most survey design models. Under the assumption of different stratum in a survey design, the distribution of the observations is not be weakly effort-homogeneous, however the distribution of observations within each stratum typically is. 

An important consequence of having a distribution in the observation equation being weakly effort homogeneous is the following:
\begin{Theorem}
Let  $g(y|E)$ denote the distribution of an observation given effort $E$, and let  $g(y|N,E = e)$ denote the distribution of an observation given effort $e$ and population size $N$. If $g(y|E)$ is weakly effort-homogeneous, then  the expectation and the variance of $Y\sim g(y|E=e)$ is a linear function of $e$. Further, the variance of the random variable $\frac{Y}{e}$ is completely determined by $\V\left[  \E\left[ Y \,| N, E=1 \right]\right]$, as $e \rightarrow \infty$.
\end{Theorem}
\begin{proof}
The first part follows immediately from the definition. To prove the second statement note that 
$
\V\left[\frac{Y}{e}| E=e \right]  =\V\left[  \E\left[ \frac{Y}{e} \,| N,E=e \right]\right] + \E\left[  \V\left[ \frac{Y}{e} \,| N,E=e \right]\right].
$
Using the first statement, the first term equals $\V\left[ E\left[ Y \,| N,E=1 \right]\right]$, and the second term equals
$
\frac{1}{e} \E\left[ \V \left[ Y  \,| N, E=1 \right]\right],
$
which goes to zero as $e \rightarrow \infty$.
\end{proof}
The theorem has the following implication for models where the observation distribution are weakly effort-homogeneous: If the effort is large, then all the variation of the observations is completely explained by the population abundance, $N_t$. This can have large consequences for misspecified model; for example assume that the observation equation comes from a, simple, capture-recapture model, and that one \emph{incorrectly} assumes that the capture probability is constant across years. Then the yearly variation in the number of captured caused by varying capture probability is incorrectly moved to higher yearly variation of $N_t$, thus it leads to an overestimation of the population's variability.

\section{Model}
This section, we build a state space model that will be applied in the result section below. The model describes a male and female population jointly.
It is common in ecology to observe populations at irregular occasions, so in the model the population needs to be defined continuously.  To address this, we use a geometric Brownian motion to model the population size.  We, of course, tailor the model to our specific example, however, the specifics --like time of harvest or time between observations-- can obviously easily be altered to fit other data sets.

\subsection{State equation}

Since we model the male and female population jointly, the latent process is defined as a bivariate vector $\mv{N}_t = [ N^F_t,N_t^M ]$.  In our application  $\mv{N}_t$ represent the population prior to the hunting season, when the index data are recorded, and $\mv{N}_{t+1/2}$ represents the population after hunting season, when the survey data are recorded. This could of course be generalized to observations at any time points. Note that the index time doesn't represent actual time; the period between $\mv{N}_t$ and $\mv{N}_{t+1/2}$ is approximately $3$ months whereas the period between $\mv{N}_{t+1/2}$  and $\mv{N}_{t+1}$ is approximately $9$ months, in our example. 
The breeding season occurs within the period $[t-1/2, t]$, and is incorporated in the model by the following state equation:
\begin{align}
\label{eq:state1}
N^F_{t} | \, N^F_{t-1/2} &\sim  \plN \left( \log( r N^F_{t-1/2} ) , \, 0.75 \sigma_F^2  \right), \\
N^M_{t} \, | N^F_{t-1/2}, N^M_{t-1/2}  &\sim  \plN \left( \log( N^M_{t-1/2} + rN^F_{t-1/2}),\,  0.75  \sigma_M^2  \right), \nonumber
\end{align}
where $|$ denotes conditioning on, $\sim$ denotes equal in distribution, and $\plN$ denotes log normal distribution, thus if $x\sim \plN(log(\mu),\sigma^2)$ then $f(x; \mu,\sigma^2) = \frac{1}{x\sigma\sqrt{2\pi}} e^{- \frac{(log(x) - \mu)^2}{2\sigma^2}}$.
The expected increase due to breeding is determined solely by the female population and is controlled by the recruitment rate $r$. The parameter $\sigma_F^2$, and $\sigma_M^2$ represent the yearly variance of the female and male populations, respectively. Since it is approximately nine months between $t-1/2$ and $t$ the variance in the state is modeled by $0.75 \sigma_F^2 $. 
 Note that the model is unstable in the sense that the population would grow to infinity if there were no hunting. We initially used a more advanced model with carrying capacity to remedy this. However, in the presence of large harvest (as in our data) and a population size far from its carrying capacity, the carrying capacity has little to no impact and thus it was removed to simplify the interpretability of the model. This more advanced model is presented in the appendix.

The  hunting season takes place within the period $[t,t +1/2]$, and the effect for the population is described by the following state equation:
\begin{align*}
N^F_{t+1/2} | N^F_{t} & \sim  \plN \left( \log( N^F_{t} -H^F_t) , 0.25 \sigma_F^2 \right) ,\\
N^M_{t+1/2} | N^M_{t} & \sim  \plN \left( \log( N^M_{t} -H^M_t) , 0.25 \sigma_M^2 \right),
\end{align*}
where $\mv{H}_t=[H_t^F,H_t^M]$ represent the number of animals killed, which is known.

\subsection{Observation equation}
For the observation data in this article, there are two distinct observation equations. First, there are unbiased survey observations from a distance sampling or capture-recapture procedure \cite{buckland1993distance}. Ideally one would incorporate the distance sampling model into the observations equation. But, since we only have access to the mean and variance of the population estimates from the distance sampling, so we define the following state space model:
\begin{align*}
Y^F_{t,1} | N^F_{t+1/2} &\sim  \plN \left( \log(N^F_{t+1/2}), \sigma^2_{t,D} \right) , \\
Y^M_{t,1} | N^M_{t+1/2}&\sim  \plN \left( \log)N^M_{t+1/2}), \sigma^2_{t,D} \right).
\end{align*}

Second, the index data is described by the following equations:
\begin{align*}
Y^F_{t,2}| N^F_{t} &\sim  Po(a_tE_tN^F_{t}), \\
Y^M_{t,2}| N^M_{t} &\sim  Po(a_tE_tN^M_{t}), 
\end{align*}
where $Po$ denotes the Poisson distribution. Here $E_t$ is the effort spent collecting the observations, and $a_t$ the bias factor, denoted  countability. This factor is closely related to observability, however we want to emphasize that it is not only a factor due to observation but a general factor coming from biased measurement. Note that the Poisson distribution above is effort-homogeneous.
\subsection{Temporal countability}

Typically,  it is assumed that countability, $a_t$, is constant across years.
 It is not hard to find situations where this assumption is unrealistic. To incorporate varying countability, we introduce a hierarchical layer in the model:
\begin{align}
\label{eq:a_t}
a_t \sim \plN( \log(\bar{a}), \sigma_a^2).
\end{align}
This layer creates an overdispersion for the distribution of observations similar to the distributions in \cite{knape2011observation}.

Since the Poisson distribution is effort-homogeneous, if the effort is large the variation between $\frac{Y^F_{t,2}}{E_t}$ and $\frac{Y^F_{t-1,2}}{E_{t-1}}$ has two sources: the countability, $a_t$, variation and the variation in the population, $N^F_t$ (the same relations obviously apply for the both female and male populations). From a management point view it is fundamental to know what causes the variation in index $\{Y_t \}_{t=1}^T$ in order to use the index data in managing the population.
The larger effect the population dynamics has on the variability of the index data, compared to the countability variation, the better one can manage the population using solely the index data. In this sense the countability can be interpreted as observation error on the state space model.

\subsection{Priors}
The choice of priors often has a large effect on the posterior distributions in a Bayesian hierarchical model, and the hierarchical framework makes it hard to formulate reasonable and interpretable priors. To address these issues we adapt the framework of \cite{martins2014penalising}.

As mentioned above it is fundamental to know what causes the variability for the observations, $\mv{Y}^F_t$. In the proposed model this variation is controlled by the parameters $\sigma_F,\sigma_M$ and $\sigma_a$. 
For instance, if $\sigma_F = 0$, then the state space equation for $N^F_{t}$ is a deterministic equation, implying that  knowing $N^F_1$ we can perfectly predict any future $N^F_t$.   The larger $\sigma_F$ is the more  $N^F_{t}$ is allowed to deviate from the deterministic path.
If $\sigma_a=0$, means that all the variability of $Y_t$ is caused by  the population dynamics. Larger$\sigma_a$  implies less information of the population size in $Y^F_t$. 

If no information is available about a parameter in the model, we wish to set the priors for the parameter so that more complex models are punished in favor of the simpler alternative. For instance, for the parameters $\sigma_F,\sigma_M$ and $\sigma_a$, a $\Gamma(\alpha,\beta)$ prior, with $\alpha \leq 1$, could be used. This would  punish more complex models, since the simplest model occurs when the parameters are zero and the larger the parameter values is the more variables the model has, see \cite{martins2014penalising}. However, this is neither a reasonable nor an easily interpretable joint prior (the product of the independent priors). For instance it is unreasonable that $\sigma_F$ and $\sigma_M$ are independent, nor is it easy to see how the choice of the independent priors affect each other.

To address this we use the following prior:
\begin{align*}
\tau  &\sim \Gamma(1,\beta_{\tau}), \,  \omega  \sim \proper{B}(\alpha_B,\beta_B), \, \nu     \sim \Gamma(\alpha_{\nu}, \beta_{\nu}),\\
\sigma_F &=  \omega \tau\, , \sigma_M =  \omega \tau \nu\, , \sigma_a = (1-\omega) \tau,
\end{align*}
where $ \proper{B}$ denotes the Beta distribution. At first glance, it appears that we have constructed a very complicated prior. However each component has a clear interpretable effect on the joint prior. The parameter $\tau$ could loosely be thought of as controlling the complexity of the model -- flexibility for $\mv{N}_t$-- since, if $\tau$ is zero the model is a deterministic model with the only variability coming from the Poisson observation equation. 

The parameter $\omega$ controls the source of the variability in the data. If $\omega$ is  one the variability in the data comes entirely from population dynamics, $\mv{N}_t$, whereas if $\omega$ is zero the variability comes entirely from observation error. Finally the parameter $\nu$ is the ratio between male and female variation. Thus if $\nu$ equals the variability of  the male and female populations is equal.

\section{Application, Local Moose Management}
The moose in Scandinavia is extensively hunted and more than 200000 animals are shot yearly. The populations are managed locally by the county board, hunters and forestry companies. The large involvement of the latter is because the moose severely damage young pine stands, causing substantial economic loss for  forestry. The latter two interests have opposing objectives: industry can only accept a population below a maximum size, while hunters only accept a population above a minimum size. 

The data, in the example below, comes from two moose management areas in the south of Sweden between the years 2000 and 2013 (Gunnilbo area: 977 km$^2$, position: 59$\textdegree$ 48'N, 15$\textdegree$ 51'E and Tiveden area: 851 km$^2$; position: 58$\textdegree$79'N, 14$\textdegree$ 55'E). For both these moose populations there are two important periods during the year:
 calving, which takes place at the end of May, and hunting season, which starts in October and ends in February. Most animals are shot during the first weeks of the season. In Sweden the hunting harvest is the main cause of moose mortality and accounts for 85-95\% of all deaths (own observations). Almost all shooting of moose is recorded with exception of illegal hunting, which is assumed to be very small. Information about the state of the population is collected during the first five hunting days, where hunters count the number of animals observed, these observations are the index data. In addition, during some years there are also unbiased estimates of the population size after the hunting season (late January-February), using either capture-recapture \cite{borchers2002estimating} or distance sampling \cite{buckland2005distance} methods.

\section{Results}
Our main objective is to estimate the population size and its dynamics. The novelty of our model is allowing for a countability that varies among years. To highlight this, we also estimate the population using a model with fixed countability, $a_t=a$ for all $t$.
Figures  \ref{fig:NGunnilbo}  and \ref{fig:NTiveden}, show the posterior median of the population size for the two models and the two areas. It is apparent that the yearly population size estimate of the fixed model is almost completely determined by $\frac{Y_{t,1}}{aE_t}$ (the circles in the figures), whereas the variable model is not that tied to the observation index. As a result of this $\sigma$, the parameter defining the populations variability, is larger for the fixed model compared to the variable model, see Figure \ref{fig:params}.

Recall, that the parameter $\omega$ controls from which source the variability of the index data comes, either  process dynamics or observation error. The posterior distribution of $\omega$ (using a uniform prior)  in Figure \ref{fig:params} shows that the data implies that a large portion of the variance is explained by the observation error, variability of $a_t$, for both areas. Especially for Gunnilbo, where the major part of the variance is explained by observation error, as is apparent since most of the mass of the posterior distribution is above 0.5. Thus the data gives little support for using a fixed countability.

The posterior distributions of the mean countability, $\bar{a}$ and $a$, for the variable and fixed models respectively, are shown in Figure \ref{fig:as}. Notice that the distributions from Gunnilbo and Tiveden are almost disjunct. Thus the countability is completely different for the two study areas, indicating strong spatial variation for countability. In studying the posterior means of $a_t$ (no figure) for the variable model we can not find a significant temporal trend, however there seem to be a common annual variability factor for both areas, in that there is significant correlation for the $a_t$s  between the two areas (p=0.048, Kendell's tau).

To highlight the higher variability in estimated population size for the fixed compared to variable model, we studied three management situations and what action the two models would suggest. Assume that the female population size of Gunnilbo, $N^F_t$, is five hundred animals and to manage next year's population one can propose a harvest this year of $H_t$ females. Figure \ref{fig:HNGunnilbo} displays the relation between the year's population size and the chosen harvest. To interpret the figure we compare the outcome of three different strategies:
\begin{enumerate}
\item Stable strategy:  The goal is set so that the median of the population should be 500, in that case both the variable and the fixed model suggest that one should harvest around 170 animals. 
\item Hunter-biased strategy: The goal is set so that the population size should not be below 500 with $90\%$ probability. Here the two models give different results: according to the fixed model the harvest can not be more then 50 animals, while the corresponding number is 110 for the variable model. That is $10\%$ versus $22\%$ of the population at time $t$
\item Forestry-biased strategy: The goal is set so that the population size should   be under 500 with $90\%$ probability. Again the two models give different results, according to the fixed model the harvest must be at least 310 animals, while the corresponding number is 250 for the variable model. That is $62\%$ versus $50\%$ of the population at time $t$.
\end{enumerate}

%

\begin{figure}
\begin{center}
\includegraphics[scale=0.45]{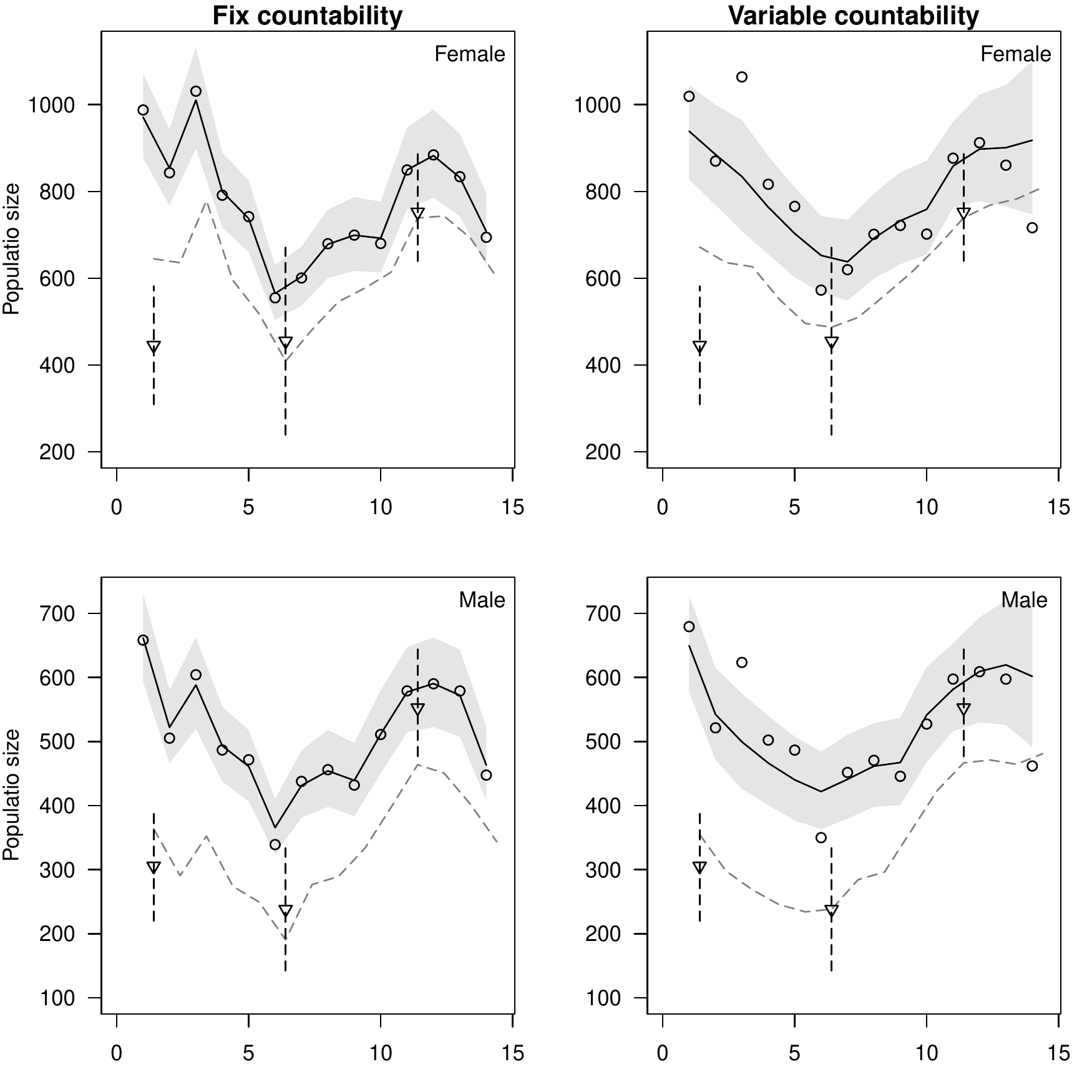}
\end{center}
\caption{Population estimate for the Gunnilbo moose population. The left figures display the estimates for fixed $a$, and the right display the estimates for a variable $a_t$. The solid line is the population before the hunt, where the shaded area is $95\%$ confidence area; the dashed line is the mean population after the hunt; the triangles are the unbiased surveys with $95\%$ confidence interval; and the circles are the observations adjusted with effort times either $\bar{a}$ or $a$.}
\label{fig:NGunnilbo}
\end{figure}

\begin{figure}
\begin{center}
\includegraphics[scale=0.45]{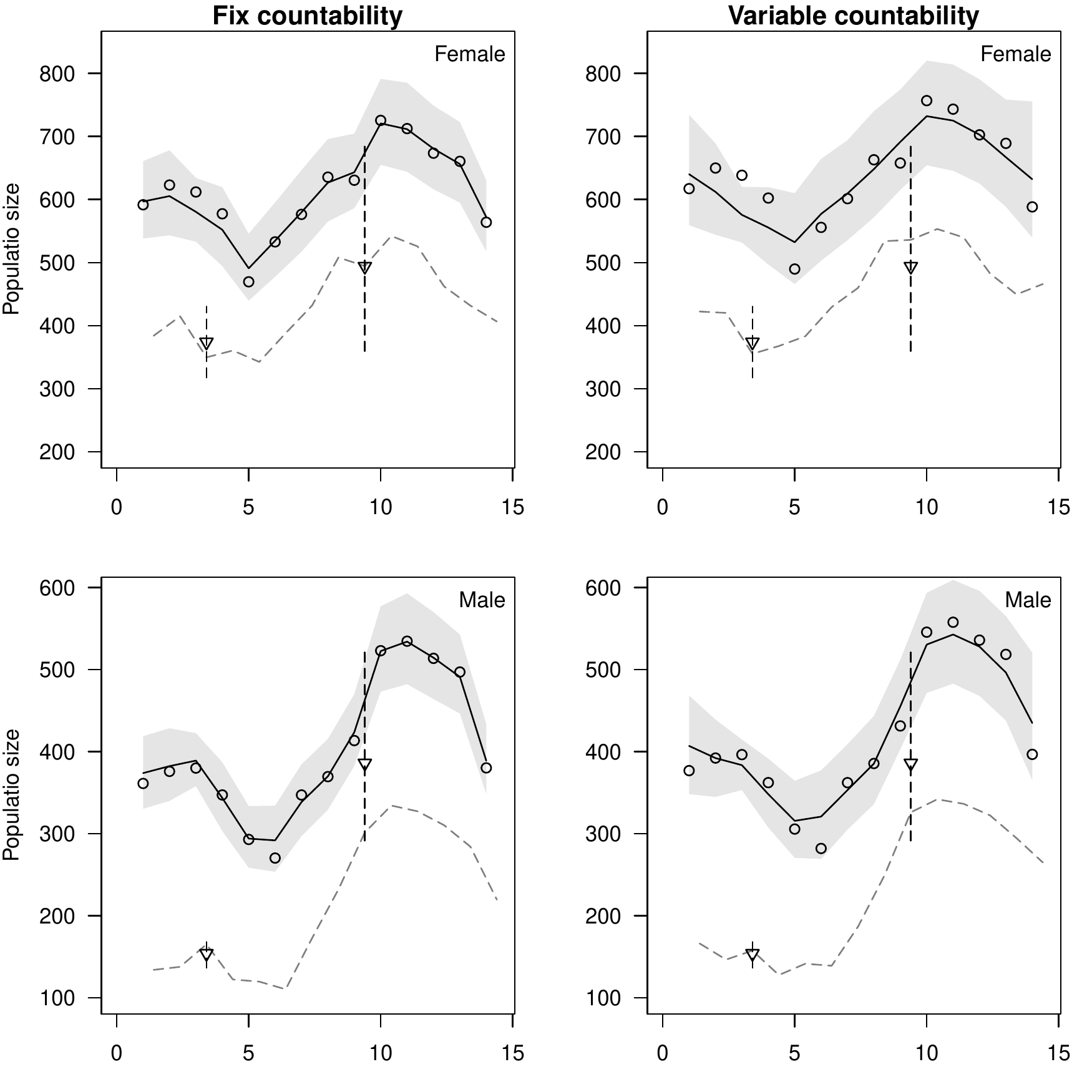}
\end{center}
\caption{Population estimates for the moose population at Tiveden. The legends are the same as in figure \ref{fig:NTiveden}}
\label{fig:NTiveden}
\end{figure}

\begin{figure}
\begin{center}
\includegraphics[scale=0.2]{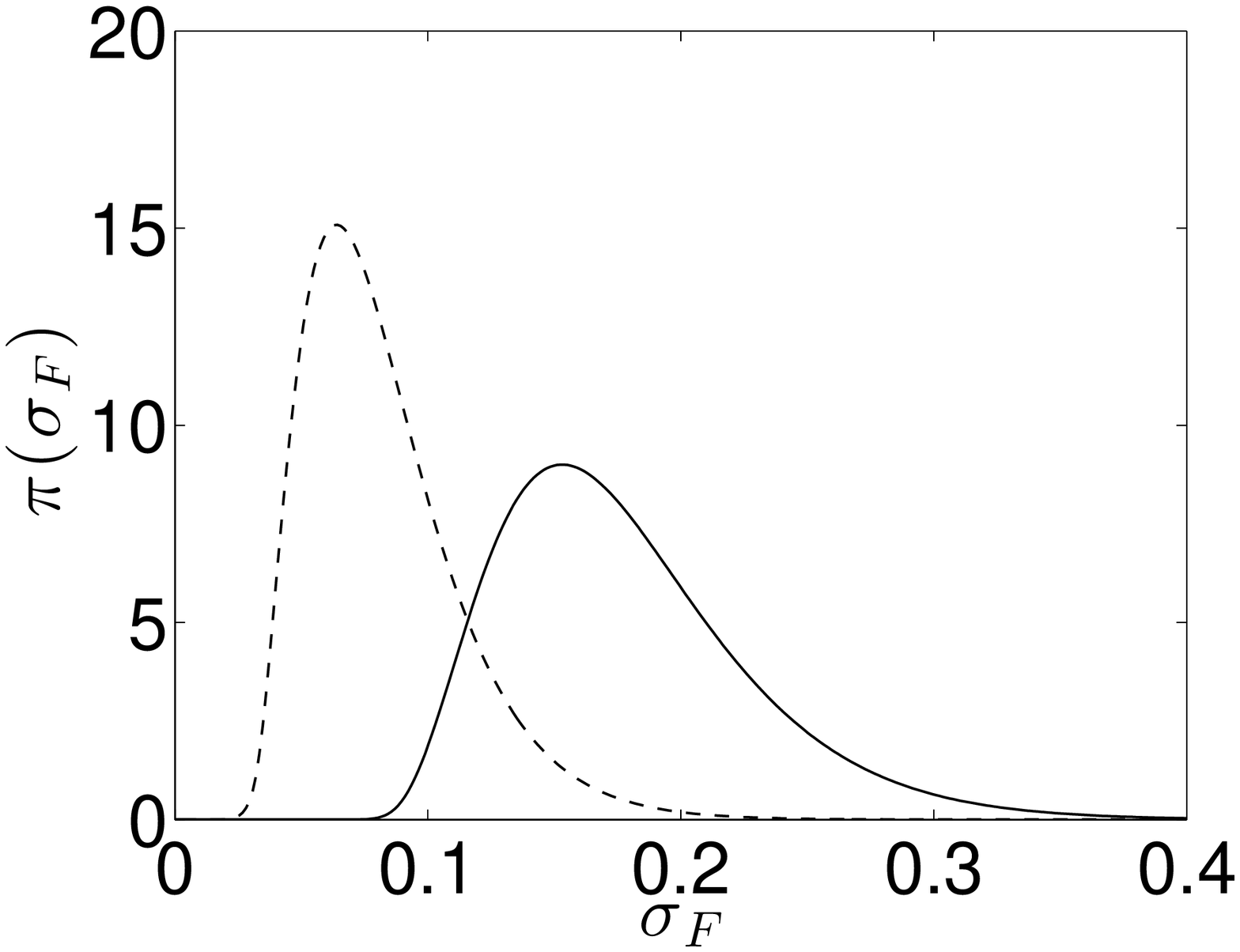}
\includegraphics[scale=0.2]{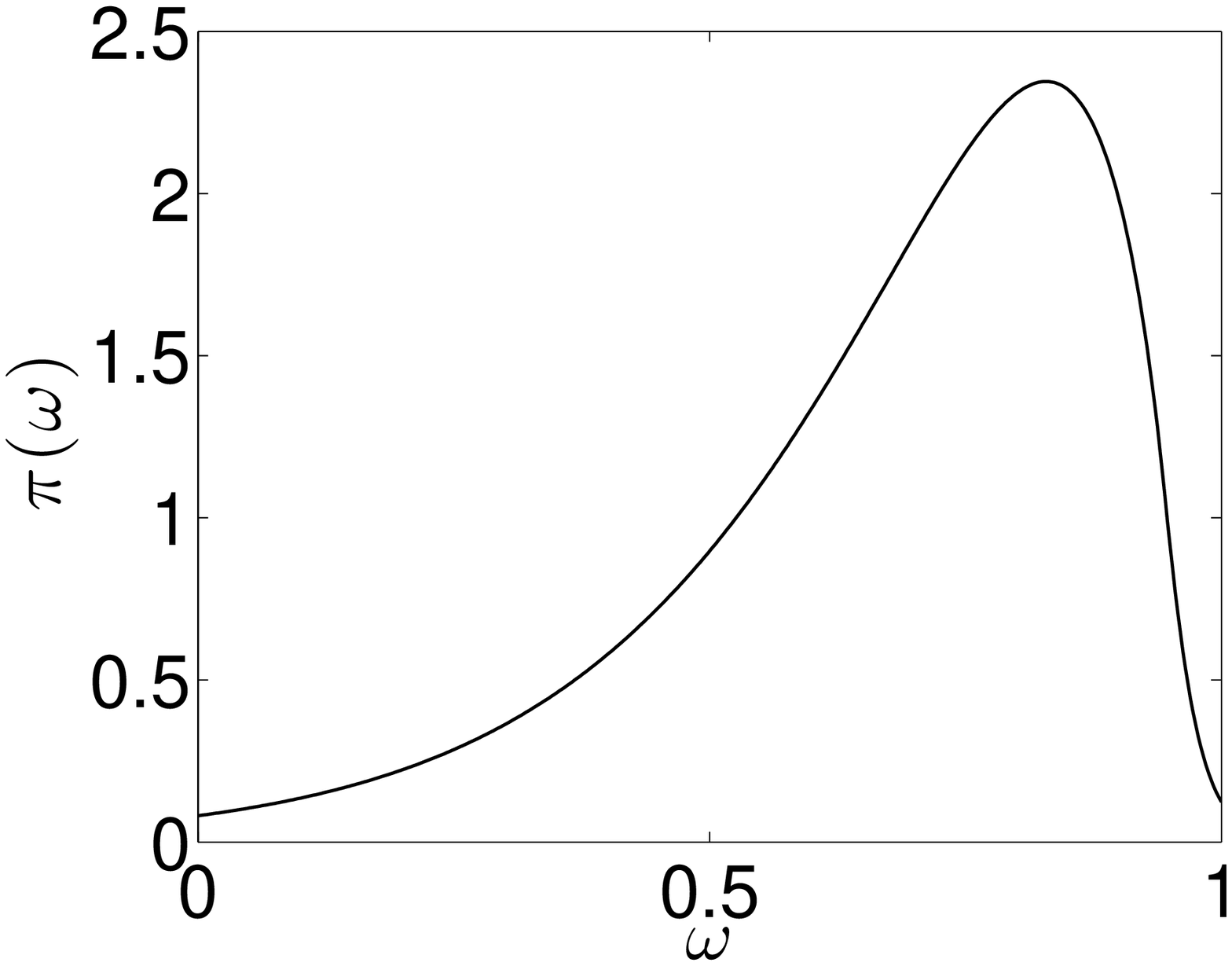}\\
\includegraphics[scale=0.2]{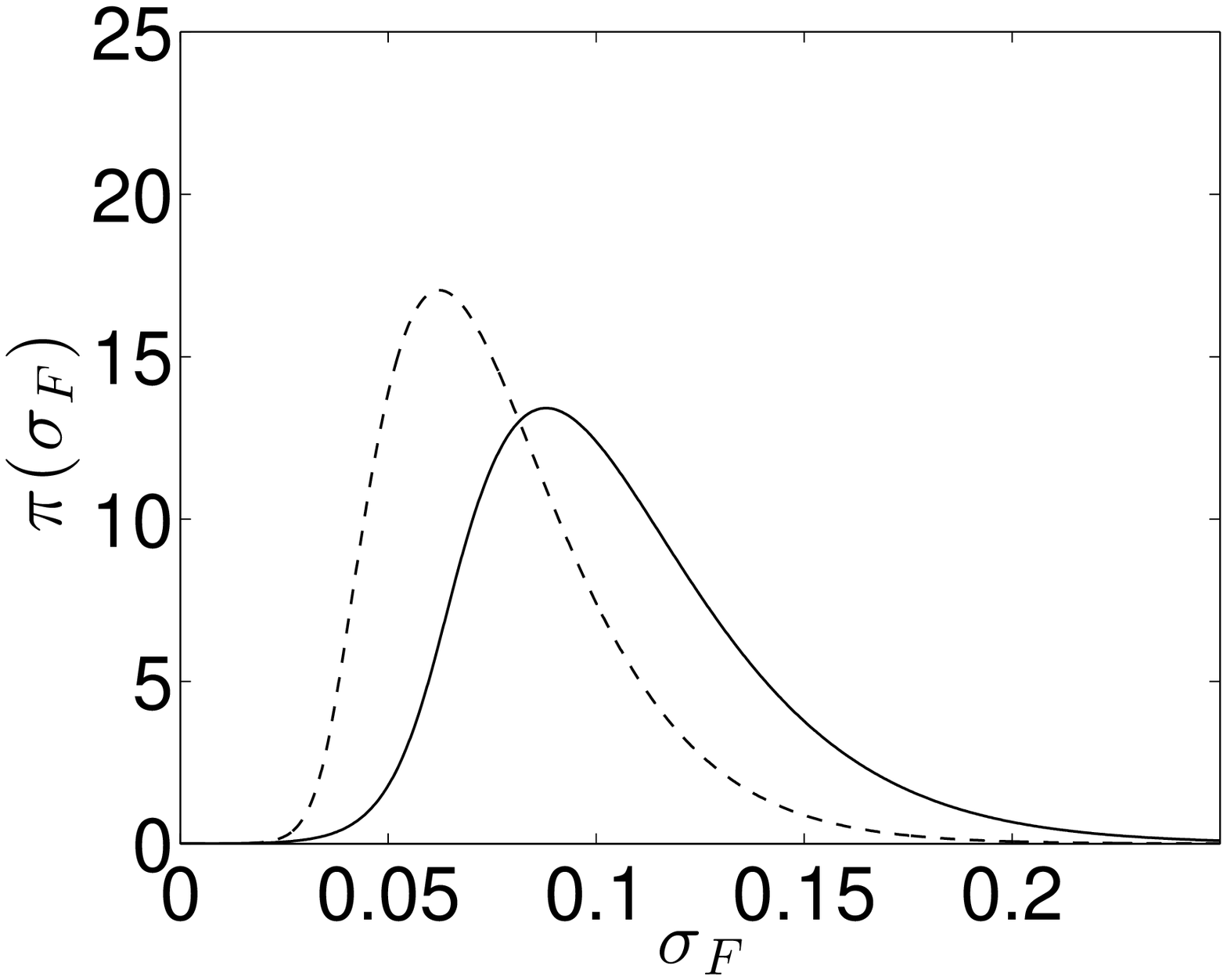}
\includegraphics[scale=0.2]{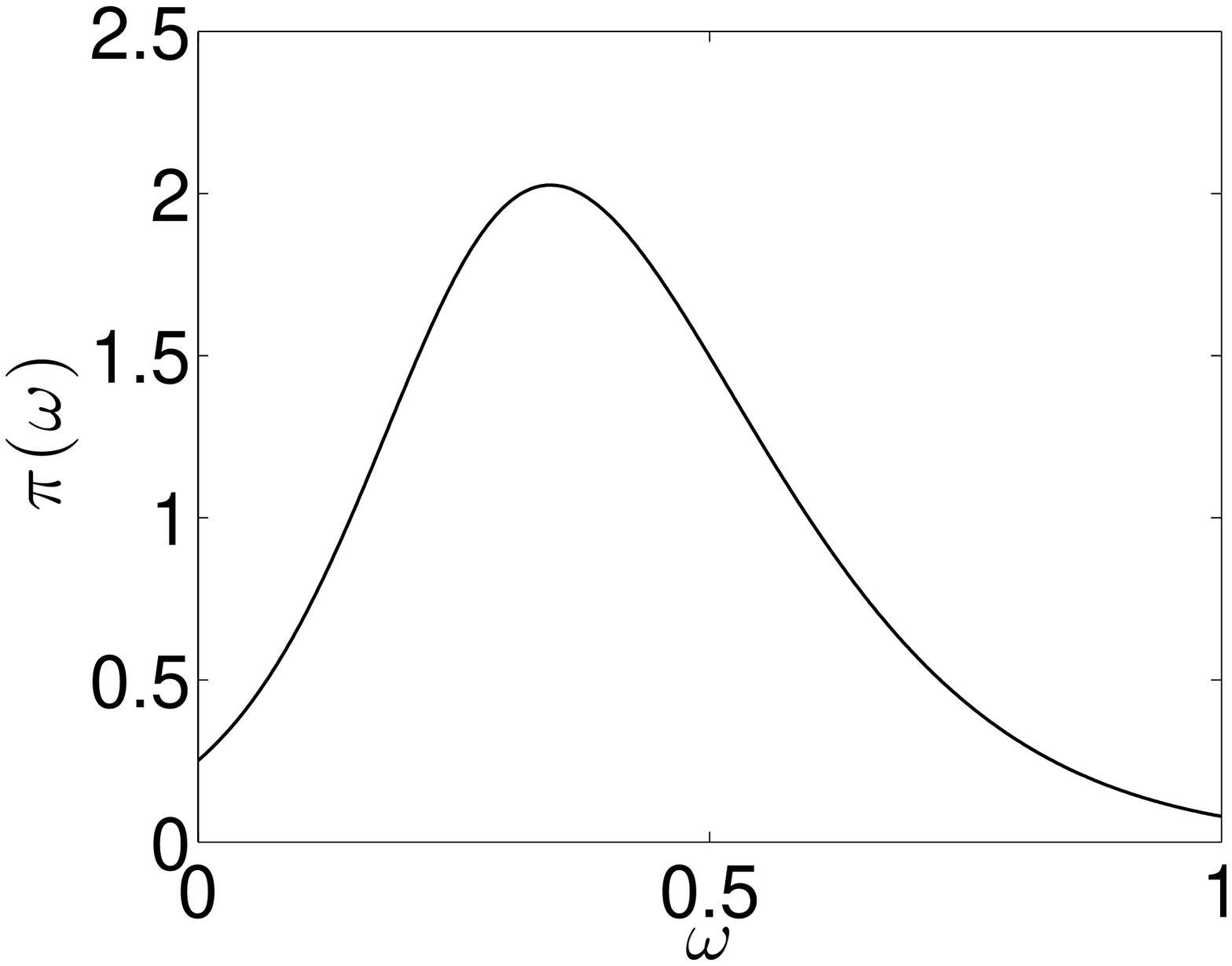}
\end{center}
\caption{The two bottom figures are from Tiveden and the two top figures are from Gunniblo.
Figures to the left displays the posterior distribution of $\sigma_F$ with varying $a$ dashed line and fixed $a$ solid line. 
Figures to the right displays the posterior distribution of $w$ thus according to the data majority of the variance is explained by the processes.}
\label{fig:params}
\end{figure}
\begin{figure}[h!]
\begin{center}
\includegraphics[scale=0.22]{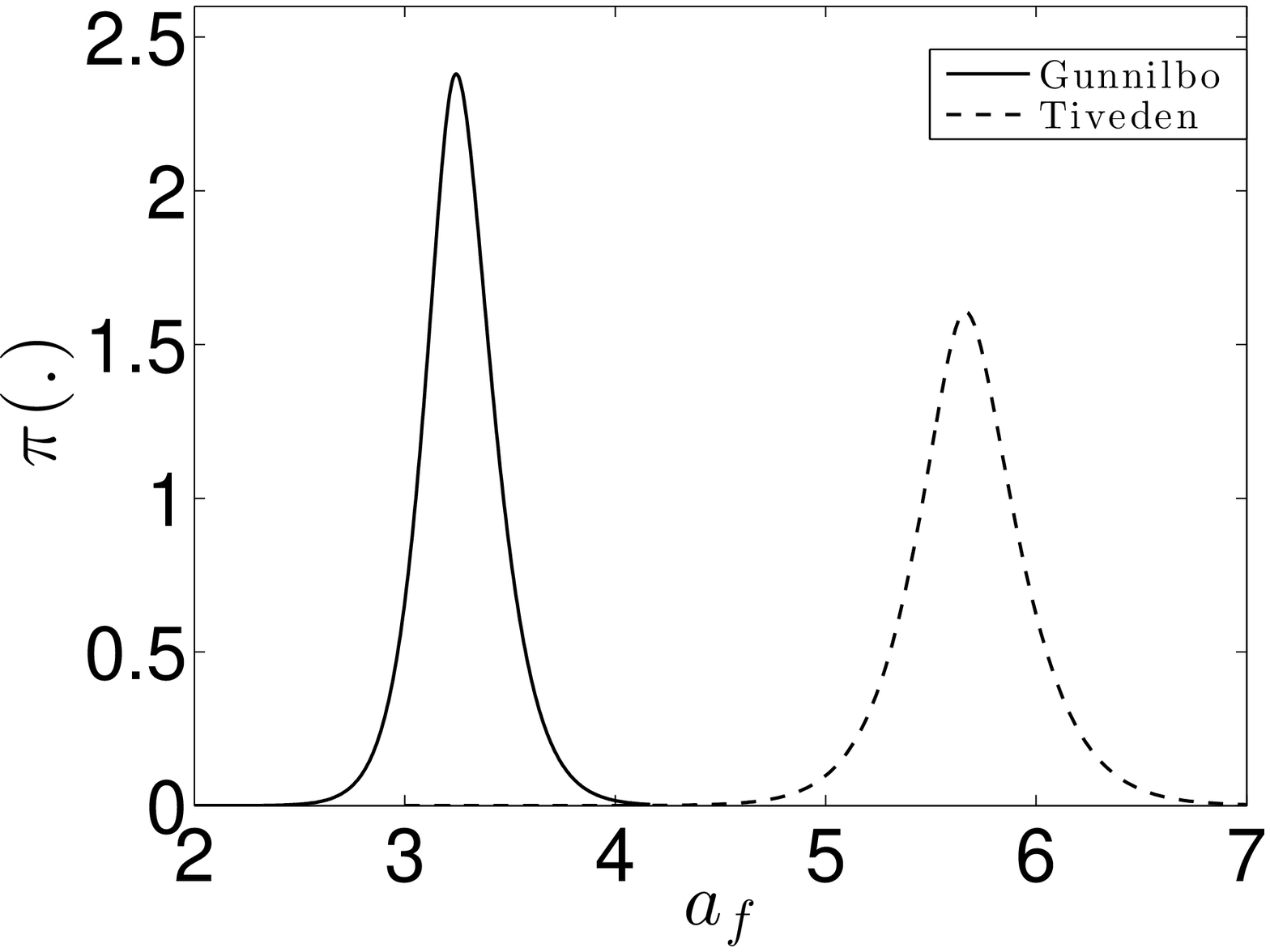}
\includegraphics[scale=0.22]{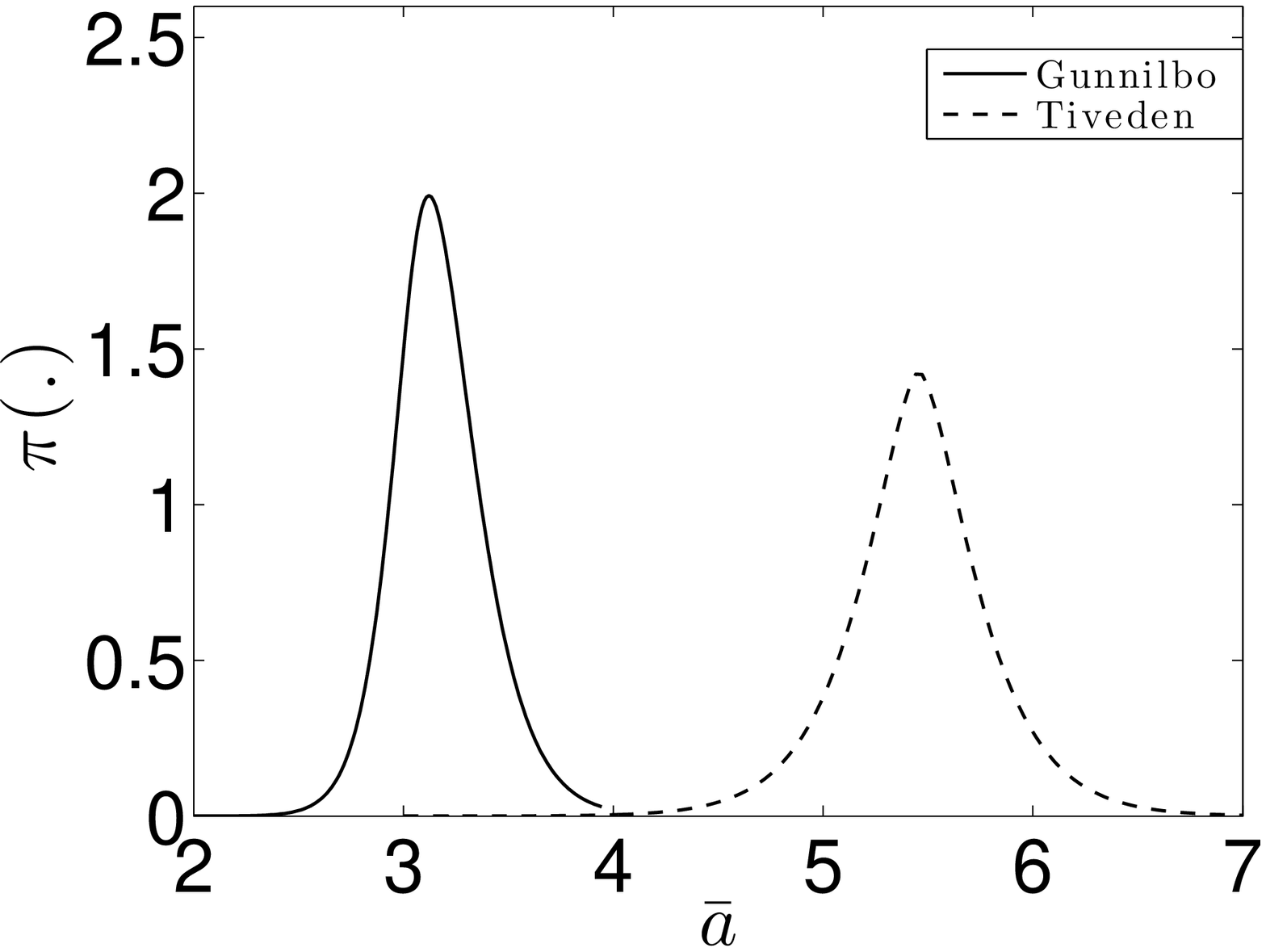} 
\end{center}
\caption{The figures display the posterior distributions for the fixed bias $a_f$,  and $\bar{a}$ for Gunnilbo and Tiveden. 
Both parameters control the expected countability for the indexes for the fixed and variable models. It is important to note that there is a scale difference between the models so one could not use the countability in Gunnilbo for Tiveden and vice versa. }
\label{fig:as}
\end{figure}

\begin{figure}[h!]
\begin{center}
\includegraphics[scale=0.22]{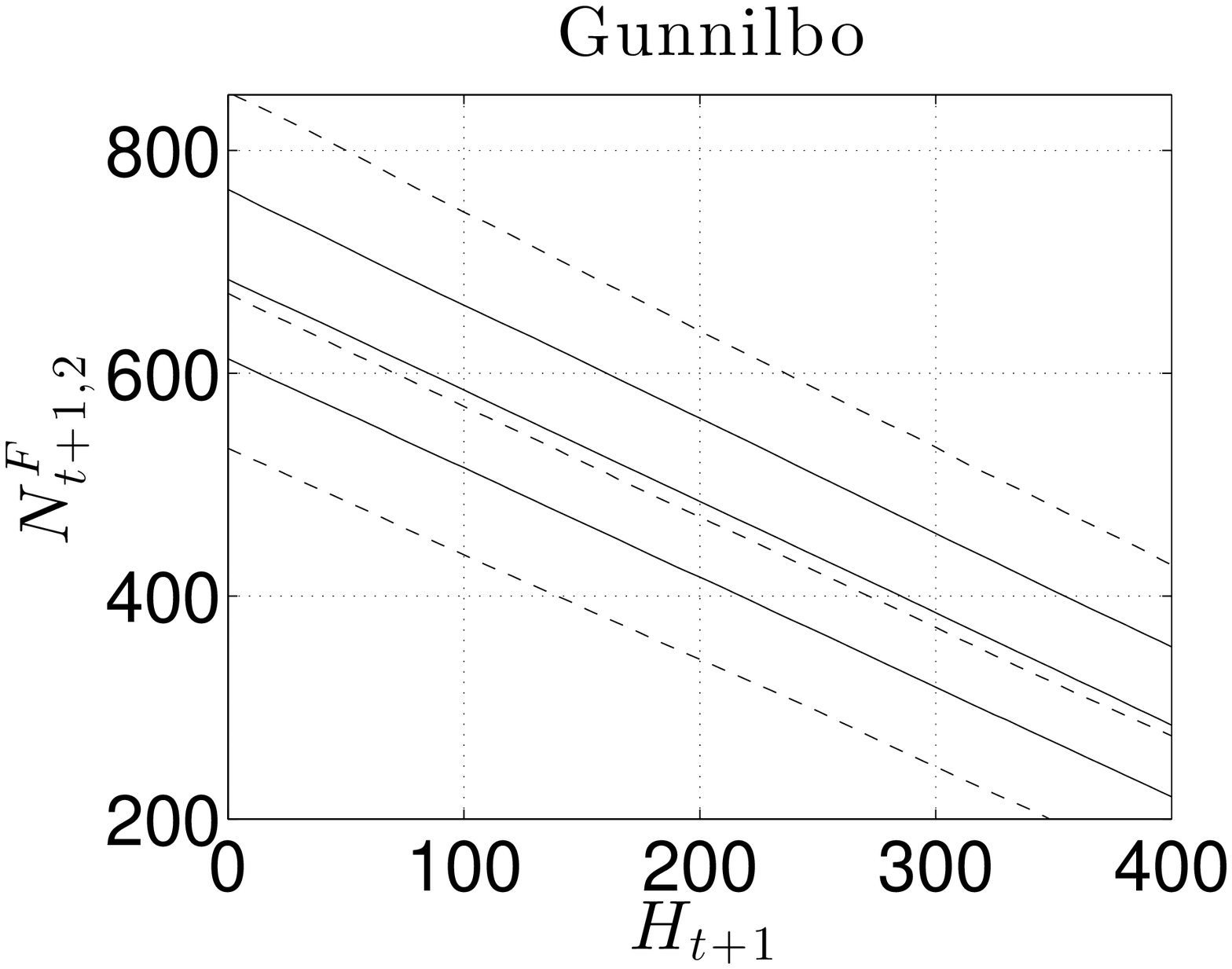}
\includegraphics[scale=0.22]{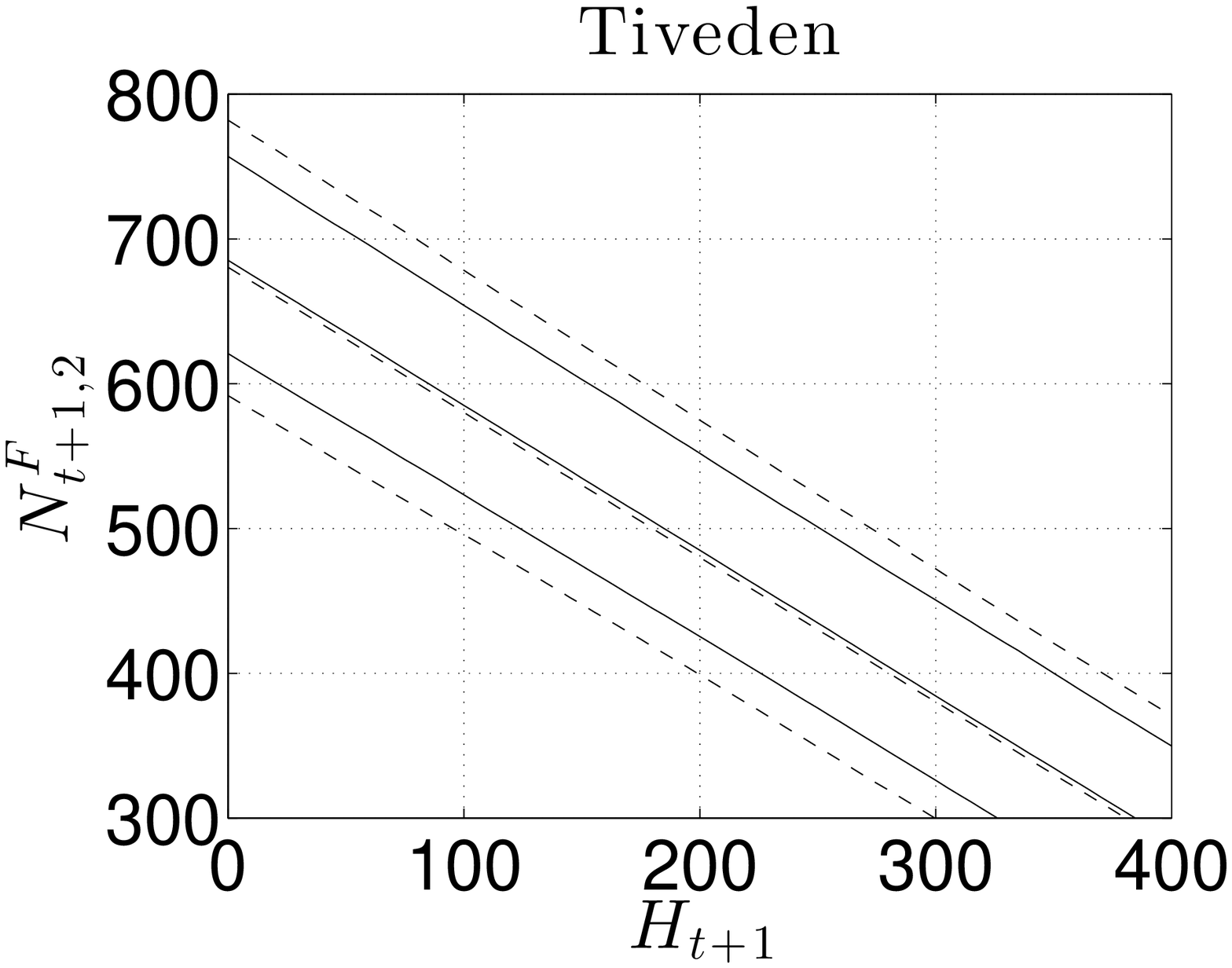} 
\end{center}
\caption{The graphs displays the effect of varying harvest on next years population ($N^F_{t+1}$), assuming a female population of 500 ($N^F_t$). The solid lines are for the variable countability model and dashed are for the fixed countability model. The three lines for each model represent the
$10\%$; $50\%$; and $90\%$ quantiles.}
\label{fig:HNGunnilbo}
\end{figure}

\section{Discussion}
Efficient population management require good information, and population size is one of the most important data one can have for this purpose. This is also true for scientific studies of populations dynamic. The difficulty of getting this information comes stems from the effort needed to observe organisms as well as the cost of applying unbiased methods. This is especially challenging when dealing with local management, where one typically only has access to index data and, at best, infrequent unbiased population estimates. Direct use of  a population size index, always risks overlooking the fact that the observed variation, in the data, is caused by variation in countability, rather then variation in the underlying population size.

The obvious step to achieve the desired population size estimates is to apply a constant countability to calibrate the index. However, as we have shown, this forces the variability in the index data to be explained solely by the population dynamics; resulting in unreliable estimation of population size and its dynamics.  To address this, we proposed a model with a temporally varying countability. The importance of the variability was illustrated by different management objectives for local moose populations and their resulting harvesting strategies. We compared a model with fixed countability to a model with variable countability; the larger population variability in the fixed model resulted in more extreme harvesting strategies.

We want to stress the importance of identifying the sources of variability in population size measurement. In our proposed model there are two potential sources of variation, the variable countability and the population dynamics. Of course, there are other sources, not accounted for like varying growth rate, environmental changes, immigration/emigration etc.  It is not possible to distinguish these sources by population size measurements alone. To analysis more complex situations additional information is needed.  For instance, in our example, if one had access to the weather conditions and the animal counts by day, this could, potentially, be used to explain some of the variation in countability. More challenging is to make inferences about unobserved population processes, for instance demographic structures, immigration/emigration etc. Finding alternative sources of data concerning these characteristics, as well as adding them into the population models possess a challenging, and important, research area in ecology. 
In summary, blindly accepting index data in modeling, can, transfer variability from the countability to population size estimate. This can result in severe consequences for management of animals. By allowing for variable countability, one gets more stable population dynamic.
\section*{Acknowledgements}
The authors would like to thank David Winkler for valuable comments on the manuscript. We are also thankful to Emil Broman and Johan Truv\'e at the Svensk Naturf\"{o}rvaltning AB for providing the moose data. Jonas Wallin has been supported by the Knut and Alice Wallenberg foundation .
\bibliographystyle{plain}
\bibliography{../algtid}
\appendix
\section{Online: Log likelihood}
Here we derive the log likelihood for a more general model than presented in the article (which corresponds to setting $K=0$):
\begin{align*}
N^F_{t} | \, N^F_{t-1/2} &\sim  \plN \left( \log( (r_{t} N^F_{t-1/2} ) , \, 0.75 \sigma_F^2  \right), \\
N^M_{t} \, | N^F_{t-1/2}, N^M_{t-1/2}  &\sim  \plN \left( \log( N^M_{t-1/2} + r_t N^F_{t-1/2}),\,  0.75  \sigma_M^2  \right),
\end{align*}
 where $r_{t} = e^{r- K\log ((N^F_{t-1/2})}$. Note that $K$  is not actually a carrying capacity since it only affects the recruitment rate. This state model is close to an Ornstein-Uhlenbeck (OU) type process, \cite{engen2007heterogeneous}.

Now we derive the log-likelihood for the variables of interest which are $\{\mv{N},\mv{a},\bar{a}, \tau, \omega, \nu, r, K\}$. We split the log likelihood into several parts to simplify understanding
$$
l = l_1 + l_2 + l_3 + l_4 + l_5.
$$
The first three parts come from the latent model, the fourth from the observation equations, and the last component from the prior.
Here 
\begin{align*}
l_1 = &- \frac{T}{2}\left( \log(\sigma_M) + \log(\sigma_F) \right) \\ 
&- 2\sum_{t=1}^{T}  \frac{1}{ \sigma^2_F}\left( \log(N^F_{t} )  - \log \left( (r_t + 1)N^F_{t-\frac{1}{2}} \right)  \right)^2
 + \frac{1}{ \sigma^2_M}\left( \log(N^M_{t} )  - \log \left( N^M_{t - \frac{1}{2}} + r_t N^F_{t-\frac{1}{2}} \right)  \right)^2 \\
 &- \sum_{t=1}^{T}  \log(N^F_{t} )  + \log(N^M_{t} ),
\end{align*}
where $r_t =  \exp \left( r - K \log(N^F_{t-\frac{1}{2}})  \right), \sigma_F= \omega \tau, \sigma_M = \omega ( 1- \tau)$.

\begin{align*}
l_2 = &- \frac{T}{2}\left( \log(\sigma_M) + \log(\sigma_F) \right) \\ 
&- \frac{4}{3}\sum_{t=1}^{T}  \frac{1}{ \sigma^2_F}\left( \log(N^F_{t+\frac{1}{2}} )  - \log \left( N^F_{t} -H^F_t \right)  \right)^2
 + \frac{1}{ \sigma^2_M}\left( \log(N^M_{t+\frac{1}{2}} )  - \log \left( N^M_{t} - H^M_t \right)  \right)^2 \\
 &- \sum_{t=1}^{T}  \log(N^F_{t+\frac{1}{2}} )  + \log(N^M_{t+\frac{1}{2}} ).
\end{align*}

Here the distribution of $N^F_{t},N^M_{t}$ above equation is only well defined if $N^F_{t} \geq H^F_t $ and  $N^M_{t} \geq H^M_t$, thus we have constrained distribution and thus must add the extra term to the log likelihood
\begin{align*}
l_3 = &-  \sum_{t=1}^{T} \log \left(  \Phi \left( ( \log \left( (r_t + 1)N^F_{t-\frac{1}{2}} \right) - \log(H^F_t)  \right) ) \sqrt{0.25}\sigma_F \right) \\
 &-  \sum_{t=1}^{T} \log \left(  \Phi \left( (   \log \left( N^M_{t - \frac{1}{2}} + r_t N^F_{t-\frac{1}{2}} \right)  - \log(H^M_t)  \right) ) \sqrt{0.25}\sigma_M \right).
\end{align*}
Here $\Phi$ is cumulative distribution function for a standard normal distribution.
\begin{align*}
l_4 =& \sum_{t=1}^{T} Y^F_{t,1} \log( a_t N^F_t) - a_t N^F_t E_t + Y^M_{t,1} \log( a_t N^M_t) - a_t N^M_t E_t \\
& - \frac{1}{2} \sum_{t \in T^*} \frac{1}{\sigma^2_{t,D}}  \left( \left(\log \left(N^F_{t+1/2} \right) - Y^F_{t,2} \right)^2 + \left(\log \left(N^M_{t+1/2} \right) - Y^M_{t,2} \right)^2 \right).
\end{align*}
Here $T^*$ are the location where there is Distance observations. And from the priors we have
\begin{align*}
l_5 =&  - T \log(\sigma_a) - \frac{1}{2\sigma^2_a}  \sum_{t=1}^{T}  \left(a_t - \bar{a} \right)^2 \\
 &-\beta_{\tau} \tau   +(\alpha_B - 1) \log(\omega) + (\beta_B - 1)\log \left(1 -\omega \right) + (\alpha_{\nu}-1)\log(\nu )- \beta_{\nu} \nu - \frac{  \left( \bar{a} - \mu_{\bar{a}}\right)^2 }{2\sigma^2_{\bar{a}}} \\
 &- \frac{1}{2 \sigma^2_r} \left(r - \mu_r\right)^2 - \beta_{K}  K.
\end{align*}
Here $\sigma_a = (1 -\omega) \tau$, note that $K,\tau, \nu$are constrained to be larger or equal to zero and $\omega\in [0,1]$.

\section{Online: Estimation}
To generate inference of the parameters, given data, we want to sample from the posterior distribution.  We use a Monte Carlo Markov chain (MCMC) algorithm, more precisely  the Metropolis adjusted Langevinan algorithm (MALA) \cite{grenander1994representations}, to generate draws from the posterior. Further we adapted the algorithm in the framework of adaptive MCMC (see \cite{haario2001,EAMCMC_Roberts}) to get better mixing of the chain. We first run $1e6$ iterations as a burn in, then use $2e6$ samples storing every thousand samples to generate the posterior distributions.  To ensure that the chain reached stationarity we used visual inspection of the trace plot.

\end{document}